\documentclass[journal=tosc]{iacrj}

\usepackage{amsmath,amsfonts,amssymb}
\usepackage{algpseudocodex}
\usepackage{algorithm}
\usepackage{braket}
\usepackage{multirow}
\usepackage{nicematrix}
\usepackage{threeparttable}

\newcommand{\etal}{\textit{et~al.}}

\title{Utilizing Circulant Structure to Optimize the Implementations of Linear Layers}

\addauthor[
orcid   = {0009-0008-0031-9106},
inst    = {1,2},
email   = {xubuji20@mails.ucas.ac.cn},
surname = {Xu},
]{Buji Xu}
\addauthor[
orcid   = {0000-0002-0281-1670},
inst    = {1,2},
email   = {sunxiaoming@ict.ac.cn},
surname = {Sun},
]{Xiaoming Sun}

\addaffiliation[
ror        = {0090r4d87},
department = {State Key Lab of Processors},
city       = {Beijing},
postcode   = {100190},
country    = {China},
]{Institute of Computing Technology, Chinese Academy of Sciences}
\addaffiliation[
ror        = {05qbk4x57},
department = {School of Computer Science and Technology},
city       = {Beijing},
postcode   = {100049},
country    = {China},
]{University of Chinese Academy of Science}

\begin{document}

    \maketitle

    \addkeywords{Linear Layer, Depth, XOR Counts, Quantum Circuit, AES}

    \begin{abstract}
        In this paper, we propose a novel approach for optimizing the linear layer used in symmetric cryptography.
        It is observed that these matrices often have circulant structure.
        The basic idea of this work is to utilize the property to construct a sequence of transformation matrices, which allows subsequent heuristic algorithms to find more efficient implementations.
        Our results outperform previous works for various linear layers of block ciphers.
        For Whirlwind $M_{0}$, we obtain two implementations with 159 XOR counts (8\% better than Yuan~\etal~at FSE 2025) and depth 17 (39\% better than Shi~\etal~at AsiaCrypt 2024) respectively.
        For AES MixColumn, our automated method produces a quantum circuit with depth 10, which nearly matches the manually optimized state-of-the-art result by Zhang~\etal~at IEEE TC 2024, only with 2 extra CNOTs.
    \end{abstract}

    \begin{textabstract}
        In this paper, we propose a novel approach for optimizing the linear layer used in symmetric cryptography.
        It is observed that these matrices often have circulant structure.
        The basic idea of this work is to utilize the property to construct a sequence of transformation matrices, which allows subsequent heuristic algorithms to find more efficient implementations.
        Our results outperform previous works for various linear layers of block ciphers.
        For Whirlwind $M_{0}$, we obtain two implementations with 159 XOR counts (8\% better than Yuan et al. at FSE 2025) and depth 17 (39\% better than Shi et al. at AsiaCrypt 2024) respectively.
        For AES MixColumn, our automated method produces a quantum circuit with depth 10, which nearly matches the manually optimized state-of-the-art result by Zhang et al. at IEEE TC 2024, only with 2 extra CNOTs.
    \end{textabstract}

    \section{Introduction}

    In recent years, lightweight cryptography has gained significant attention as the rapid development of resource-constrained applications, such as Radio-Frequency Identification (RFID) tags and Internet of Things (IoTs).
    Typically constraints include limited circuit size, power consumption and latency.

    Research generally follows two main paths.
    The first is the design of new ciphers using efficient components.
    Maximum Distance Separable (MDS) matrices are widely used in linear layers due to their ability to achieve the maximum branch number.
    Considerable research has been conducted on constructing lightweight MDS matrices, including \cite{FSE:SKOP15,C:BeiKraLea16,FSE:LiuSim16,FSE:LiWan16,ToSC:SarSye16,ToSC:GLGWL17,ToSC:ZhoWanSun18,ToSC:LSLWH19,DCC:YanZenWan21,zhaoConstruction4x4Lightweight2024}.
    Structures like circulant, Hadamard, Toeplitz, or involution matrices have been used to design matrices to reduce the search space and the number of XOR operations.

    The second is the optimization of existing cipher.
    Two important criteria are gate count and circuit depth (or latency).
    The BP algorithm~\cite{JC:BoyMatPer13} employs a greedy strategy on reducing distances to outputs, and can produce satisfiable results for many small-scale matrices.
    Several improvements upon the BP algorithm have been proposed in subsequent research~\cite{viscontiImprovedUpperBounds2018,banikMoreResultsShortest2019,TCHES:MaxEkd19,TCHES:TanPey19,ToSC:ShiFenXu23}.
    Xiang~\etal~\cite{ToSC:XZLBZ20} introduced a new approach to optimize in-place circuit by means of reduction rules.
    Lin~\etal~\cite{linFrameworkOptimizeImplementations2021} later adopted the idea for general circuit.

    The synthesis of linear layers is also of interest in quantum computing~\cite{nielsenQuantumComputationQuantum2000}.
    This is also known as CNOT circuit synthesis.
    CNOT gate is fundamental to quantum computing, as they facilitate entanglement, a crucial resource for quantum algorithms.
    CNOT gates play an important role in different areas of quantum computing, including quantum error correction~\cite{gottesmanStabilizerCodesQuantum1997,aaronsonImprovedSimulationStabilizer2004}, arithmetic circuits~\cite{brugiereGaussianEliminationGreedy2021,brugiereReducingDepthLinear2021} and randomized circuit benchmarking~\cite{knillRandomizedBenchmarkingQuantum2008,magesanScalableRobustRandomized2011}.
    The synthesis of CNOT citcuits has been studied by many researchers, such as~\cite{mooreParallelQuantumComputation2001,patelOptimalSynthesisLinear2008,jiangOptimalSpacedepthTradeoff2020,wuOptimizationCNOTCircuits2023}.

    The optimization of linear layers can also contribute to the quantum attack against cryptography. Grover's algorithm~\cite{groverFastQuantumMechanical1996} provides a quadratic speedup for unstructured search, which halves the effective key length of symmetric-key ciphers against exhaustive search.
    Additionally, various techniques utilizing the power of quantum computation are emerging~\cite{C:KLLN16,AC:BHNSS19,EC:HosSas20,EC:NaySch20,AC:BLNS21,EC:BonSchSib22}.
    In NIST's call for proposals for the standardization of post-quantum cryptography, the complexity of a quantum circuit for AES was chosen as a baseline to categorize security levels against quantum attacks, where the gate count within \texttt{MAXDEPTH} was proposed to be the metric.
    Kim~\etal~\cite{kimTimeSpaceComplexity2018} noted that parallelization strategies for Grover's algorithm could help address the depth constraint.
    They introduced the DW (depth times width) cost to better capture time-space trade-off, which has been widely used since its adoption in EuroCrypt 2020~\cite{EC:JNRV20}.
    To reduce the overall DW cost, it is essential to reduce the depth of linear layers, especially when low-depth implementations for non-linear component are used, like~\cite{AC:HuaSun22,AC:LPZW23,EC:HuaZhaLin25}.

    Several studies have focused on optimizing CNOT circuits.
    Xiang~\etal~\cite{ToSC:XZLBZ20} combined greedy algorithm and Gaussian elimination to reduce gate counts.
    They also utilized several reduction rules to further optimize the result.
    Brugière~\etal~\cite{brugiereReducingDepthLinear2021} developed a cost-function-guided approach to minimize circuit depth.
    Zhu and Huang~\cite{zhuOptimizingDepthQuantum2023} proposed a method to optimize the depth of existing CNOT circuits.
    Shi and Feng~\cite{AC:ShiFen24} improved upon Brugière~\etal's algorithm, achieving state-of-the-art depth for many matrices used in block ciphers, including AES MixColumn.
    Yuan~\etal~\cite{ToSC:YWSZZ24} searched for complete reduction rules within given length $K$ to further improve the result of Xiang~\etal.

    \subsection{Our Contributions}

    In this paper, we focus on linear layers that can be represented by a matrix with circulant structure.
    Circulant structure is widely adopted in the design of MDS matrices, such as AES MixColumn, and is especially favored in lightweight cryptography~\cite{FSE:LiuSim16}.
    Shi and Feng~\cite{AC:ShiFen24} managed to optimize the depth of many invertible linear transformations used in block ciphers, more than half of which are actually circulant matrices.
    However, existing methods have yet to fully exploit this property.

    We present a novel algorithm for synthesizing quantum circuits for such matrices.
    It can reduce the depth and optimize the gate count for various matrices used in cryptography.
    For example, for the linear transformation $M_0$ used in Whirlwind~\cite{DCC:BNNRT10}, compared with \cite{AC:ShiFen24}, the depth is reduced from 28 to 17 and the gate count is reduced from 286 to 200.
    For AES MixColumn, our algorithm achieves the state-of-the-art depth (10) and reduced gate count (107), which outperforms \cite{AC:ShiFen24} (depth 10, gate count 131) and closely matches the manually optimized result (depth 10, gate count 105) by Zhang~\etal~\cite{zhangOptimizedQuantumCircuit2024}. Besides, for linear transformation like Whirlwind $M_{0}$, our algorithm can also be used to optimize the XOR count of classical circuits.

    Moreover, we give an algorithm (in Appendix) to automate the process of Zhang~\etal's manual optimization for the quantum circuit of MixColumn.

    The source code and implementations are available at \url{https://github.com/DerryAlex/circulant_synth}.

    \subsection{Organization}

    The paper is structured as follows.
    Section 2 presents relevant knowledge and backgrounds.
    Section 3 describes some existing heuristic methods for matrix decomposition.
    Our new algorithm for circulant matrix is introduce in Section 4.
    At last, Section 5 discusses and concludes the paper.

    \section{Preliminaries}

    In this section, we offer some relevant knowledge and introduce the notation used in this work.

    \paragraph{CNOT Circuit}

    A qubit $\ket{\psi} = a \ket{0} + b \ket{1}$ can be described as a unit vector in $\mathbb{C}^{2}$, where $|a|^{2} + |b|^{2} = 1$.
    Common quantum gates include Hardamard gate ($H$), phase gate ($S$), Pauli-X gate ($X$) and controlled-NOT(CNOT) gate.
    \begin{gather*}
        H = \frac{1}{\sqrt{2}} \begin{pmatrix}
            1 & 1 \\
            1 & -1
        \end{pmatrix},
        S = \begin{pmatrix}
            1 & 0 \\
            0 & i
        \end{pmatrix},
        X = \begin{pmatrix}
            0 & 1 \\
            1 & 0
        \end{pmatrix},
        T = \begin{pmatrix}
            1 & 0 \\
            0 & e^{i \pi / 4}
        \end{pmatrix}, \\
        \text{CNOT} = I \otimes X = \begin{pmatrix}
            1 & 0 & 0 & 0 \\
            0 & 1 & 0 & 0 \\
            0 & 0 & 0 & 1 \\
            0 & 0 & 1 & 0
        \end{pmatrix}  .
    \end{gather*}
    CNOT gate can be viewed as an \emph{in-place} XOR gate as $\text{CNOT} \ket{x} \ket{y} = \ket{x} \ket{x \oplus y}$, which performs invertible linear map $\begin{pmatrix}
        1 & 0 \\
        1 & 1
    \end{pmatrix}$ over $\mathbb{F}_{2}$.
    Quantum circuit consists solely of CNOT gates is called as CNOT circuit.

    \paragraph{Elementary Matrix}

    The following three types of row operations are called as elementary row operations, as they are fundamental in matrix algebra.
    \begin{enumerate}
        \item Swap the position of two rows.
        \item Multiply one row by a nonzero scalar.
        \item Add a scalar multiple of one row to another row.
    \end{enumerate}

    The linear layer of a block cipher is actually an $n$-bit reversible linear boolean function, which is equivalent to an invertible matrix $A$ in $GL(\mathbb{F}_{2}, n)$, where $GL(\mathbb{F}, n)$ is the general linear group of invertible $n \times n$ matrices over field $\mathbb{F}$.
    The Gauss-Jordan elimination can be used to transform any invertible matrix into an identity matrix by performing a series of elementary operations.
    Formally, we have the following theorem.
    Readers can find proof in textbooks, such as~\cite{artinAlgebra2018}.

    \begin{theorem}
        For any field $\mathbb{F}$, every invertible matrix $A$ in $GL(\mathbb{F}, n)$ can be transformed into identity matrix using elementary row operations.
    \end{theorem}

    A matrix obtained by performing an elementary operation on identity matrix $I_{n}$ is called an elementary matrix.
    By exchanging $i$-th and $j$-th row of an identity matrix, we obtain a type-1 elementary matrix, denoted as $E(i \leftrightarrow j)$.
    In $\mathbb{F}_{2}$, the nonzero scalar must be 1.
    Hence, the second type operation becomes no-op.
    By adding $j$-th row to $i$-th row of an identity matrix, we obtain a type-3 elementary matrix, denoted as $E(i + j)$.
    The CNOT gate $\text{CNOT}(i, j)$ controlled by the i-th qubit and targeting on the j-th qubit corresponds precisely to the type-3 elementary matrix $E(j + i)$.
    Note that the type-1 elementary matrix can also be decomposed as type-3 elementary matrix: $E(i \leftrightarrow j) = E(i + j) E(j + i) E(i + j)$.
    Therefore, linear transformation can be implemented using only CNOT gates.

    \begin{example}
        Assume $n = 3$. Then,
        \begin{gather*}
            E(1 \leftrightarrow 3) = \begin{pmatrix}
                0 & 0 & 1 \\
                0 & 1 & 0 \\
                1 & 0 & 0
            \end{pmatrix},
            E(1 \leftrightarrow 3) \begin{pmatrix}
                x_{1} \\
                x_{2} \\
                x_{3}
            \end{pmatrix} = \begin{pmatrix}
                x_{3} \\
                x_{2} \\
                x_{1}
            \end{pmatrix}, \\
            E(1 + 3) = \begin{pmatrix}
                1 & 0 & 1 \\
                0 & 1 & 0 \\
                0 & 0 & 1
            \end{pmatrix},
            E(1 + 3) \begin{pmatrix}
                x_{1} \\
                x_{2} \\
                x_{3}
            \end{pmatrix} = \begin{pmatrix}
                x_{1} + x_{3} \\
                x_{2} \\
                x_{3}
            \end{pmatrix} .
        \end{gather*}
    \end{example}

    It is easy to see that the type-3 elementary matrices can be adjusted to be adjacent by the following property
    \begin{gather}
        E(i + j) E(k \leftrightarrow l) = E(k \leftrightarrow l) E(f_{k,l}(i) + f_{k,l}(j)) \label{eq:type-3-move} , \\
        E(k \leftrightarrow l) E(i + j) = E(f_{k,l}(i) + f_{k,l}(j)) E(k \leftrightarrow l) ,
    \end{gather}
    where
    \begin{equation}
        f_{k,l}(x) = \begin{cases}
            k, & \text{if } x = l , \\
            l, & \text{if } x = k , \\
            x, & \text{else}.
        \end{cases}
    \end{equation}
    Accordingly, we have the following theorem.
    \begin{theorem}
        Any $A$ in $GL(\mathbb{F}_{2}, n)$ can be decomposed as
        \begin{equation}
            A = E(i_{1} + j_{1}) \dots E(i_{s} + j_{s}) \cdot P \cdot E(i_{s + 1} + j_{s + 1}) \dots E(i_{t} + j_{t}) \label{eq:final_solution}
        \end{equation}
        where $P$ is a permutation matrix (each row or column has $(n - 1)$ zeros and one $1$).
    \end{theorem}
    Remind that $P$ can be realized for free through rewiring.
    We need to find a sequence of elementary matrices $E(i_{1} + j_{1}), \dots, E(i_{t} + j_{t})$ such that
    \begin{equation}
        E(i_{s} + j_{s}) \dots E(i_{1} + j_{1}) \cdot A \cdot E(i_{t} + j_{t}) \dots E(i_{s + 1} + j_{s + 1}) = P \label{eq:solution} .
    \end{equation}
    The equation relies on the fact that $E(i + j)$ over $\mathbb{F}_{2}$ is involution, i.e., $E(i + j)^{-1} = E(i + j)$.

    Note that by left-multiplying $E(i + j)$ to $M$, we add the $j$-th row to $i$-th row of $M$.
    By right-multiplying $E(i + j)$ to $M$, we add $i$-th column to $j$-th column of $M$.
    That is, by applying gate $CNOT(j, i)$ after quantum circuit for $M$, we perform an elementary row operation.
    By applying gate $CNOT(j, i)$ before quantum circuit for $M$, we perform an elementary column operation.

    \paragraph{Metrics}

    \cite{ToSC:XZLBZ20} used g-XOR and s-XOR to name two metrics regarding the XOR count of linear layers.
    g-XOR is the number of XOR operations to compute the outputs from given inputs, while s-XOR is the number of \emph{in-place} XOR operations to compute the outputs.
    The key difference is that g-XOR stores the result in a new variable whereas s-XOR forces in-place update.
    The two concepts are formally defined as follows.

    \begin{definition}[g-XOR]
        Given a linear matrix $M_{m \times n}$ over $\mathbb{F}_{2}$.
        The g-XOR count of $M$ is the number of operations $t_{i} = t_{j_1} \oplus t_{j_2} (1 \leq j_1, j_2 < i)$ such that the outputs $\mathbf{y} = M \mathbf{x}$ are subset of all $t_{i}$s, where $\mathbf{x}$ are the inputs and $t_{k} = x_{k}$ for $k \leq n$.
    \end{definition}

    \begin{definition}[s-XOR]
        Let $M \in GL(n, \mathbb{F}_{2})$ be an invertible matrix.
        The s-XOR count of $M$ is the number of XOR instructions $x_{i} \gets x_{i} \oplus x_{j} (1 \leq i, j \leq n)$ to update the inputs $\mathbf{x}$ to the outputs $\mathbf{y} = M \mathbf{x}$.
    \end{definition}

    \begin{example}
        Let $M$ be the matrix $\begin{pmatrix}
            1 & 0 & 0 & 0 \\
            0 & 1 & 0 & 0 \\
            1 & 1 & 1 & 0 \\
            1 & 1 & 0 & 1
        \end{pmatrix}$.
        The g-XOR count of $M$ is 3.
        \begin{equation}
            \begin{aligned}
                t_{5} &= x_{1} \oplus x_{2} \\
                t_{6} &= t_{5} \oplus x_{3} = y_{3} \\
                t_{7} &= t_{5} \oplus x_{4} = y_{4}
            \end{aligned}
        \end{equation}
        The s-XOR count of $M$ is 4.
        \begin{equation}
            \begin{aligned}
                x_{3} & \gets x_{3} \oplus x_{1} \\
                x_{3} & \gets x_{3} \oplus x_{2} \\
                x_{4} & \gets x_{4} \oplus x_{1} \\
                x_{4} & \gets x_{4} \oplus x_{2}
            \end{aligned}
        \end{equation}
    \end{example}

    \begin{definition}[Depth]
        Given an implementation $\mathcal{I}$ of $A$, the depth of a node $t$ in $\mathcal{I}$ is defined as the length of the longest path from an input node to $t$ in the implementation graph, denoted by $d(t)$.
        In particular, the depth of all input nodes is defined as 0.
        The depth of $\mathcal{I}$ is defined as the maximum depth of all output nodes denoted by $d(I)$, that is $d(\mathcal{I}) = \max_{0 \leq i < m} d(y_i)$.
    \end{definition}
    The depth of nodes follows the property:
    \begin{property}
        For three nodes $t_1, t_2, t_3$ in an implementation, if $t_1$ is generated by $t_2$ and $t_3$, then $d(t_1) = \max \{ d(t_2), d(t_3) \} + 1$.
    \end{property}

    \paragraph{Field Extension}

    A polynomial $f(x)$ in $\mathbb{F}[x]$ is said to be irreducible if there do not exist two non-constant polynomial $g(x), h(x)$ in $\mathbb{F}[x]$ such that $f(x) = g(x) h(x)$.
    For example, $x^{2} + x + 1$ is irreducible in $\mathbb{F}_{2}[x]$, but $x^{2} + x$ is not, since $(x + 1)(x + 1) = x^{2} + 2x + 1 \equiv x^{2} + 1 \pmod{2}$.
    Up to isomorphism, field with $2^{k}$ elements is equal to the polynomial ring over $\mathbb{F}_{2}$ modulo an irreducible polynomial $g(x)$ of degree $k$, $\mathbb{F}_{2^k} \cong \mathbb{F}_{2}[x] \setminus \langle g(x) \rangle$.
    For more details, readers can check Galois theory in textbook~\cite{artinAlgebra2018}.

    Elements in $\mathbb{F}_{2^k}$ can be viewed as a vector in $\mathbb{F}_{2}^{k}$.
    Then, multiplication by an element $\alpha \in \mathbb{F}_{2^k}$ can be described as a left multiplication with a matrix $T_{\alpha} \in \mathbb{F}_{2}^{k \times k}$.
    This gives an isomorphism from $\mathbb{F}_{2^k}$ to $\mathbb{F}_{2}^{k \times k}$.
    The matrix $T_{\alpha}$ can be called as the representation of $\alpha$.
    For more details, readers can check group representation theory in textbook~\cite{artinAlgebra2018}.

    \begin{example}
        Take $g(x) = x^{2} + 1$ and $\bar{1} = 1$, $\bar{2} = x$, $\bar{3} = x + 1$.
        As $\bar{2} \cdot 1 = x$, $\bar{2} \cdot x = x^{2} \equiv 1 \pmod{x^2 + 1}$, $\bar{2}$ can be represented by
        \begin{equation}
            T_{2} = \begin{pmatrix}
                0 & 1 \\
                1 & 0
            \end{pmatrix} .
        \end{equation}
        Similarly, we can calculate that
        \begin{equation}
            T_{1} = \begin{pmatrix}
                1 & 0 \\
                0 & 1
            \end{pmatrix}, T_{3} = \begin{pmatrix}
                1 & 1 \\
                1 & 1
            \end{pmatrix} .
        \end{equation}
        We can verify that $\Phi: \alpha \mapsto T_{\alpha}$ is an isomorphism from $\mathbb{F}_{2^2}$ to $\mathbb{F}_{2}^{2 \times 2}$.
        For example,
        \begin{gather}
            \Phi(\bar{1} + \bar{2}) = \Phi(\bar{1}) + \Phi(\bar{2}) = \begin{pmatrix}
                1 & 0 \\
                0 & 1
            \end{pmatrix}
            + \begin{pmatrix}
                0 & 1 \\
                1 & 0
            \end{pmatrix}
            = \begin{pmatrix}
                1 & 1 \\
                1 & 1
            \end{pmatrix} = \Phi(\bar{3}), \\
            \Phi(\bar{2} \cdot \bar{3}) = \Phi(\bar{2}) \cdot \Phi(\bar{3}) = \begin{pmatrix}
                0 & 1 \\
                1 & 0
            \end{pmatrix}
            \begin{pmatrix}
                1 & 1 \\
                1 & 1
            \end{pmatrix}
            = \begin{pmatrix}
                1 & 1 \\
                1 & 1
            \end{pmatrix} = \Phi(\bar{3}) .
        \end{gather}
    \end{example}

    \paragraph{MDS Matrix}

    The hamming weight $\mathrm{hw}_{k}(v)$ of vector $v \in \mathbb{F}_{2^k}^{n}$ is defined to be the number of nonzero entries in $v$.
    The branching number $\mathrm{bn}_{k}(M)$ of matrix $M \in \mathbb{F}_{2^k}^{n \times n}$ is defined to be $\min_{u \in \mathbb{F}_{2^k}^{n} \setminus \{0\}} \{ \mathrm{hw}_{k}(u) + \mathrm{hw}_{k}(M u) \}$.
    Matrix $M \in \mathbb{F}_{2^k}^{n \times n}$ is MDS (Maximum Distance Separable) if $\mathrm{bn}_{k}(M) = n + 1$.
    MDS matrices are commonly used in the design of linear layer of block ciphers as they help resist differential and linear attacks.

    A $n \times n$ matrix $M$ is circulant if there exists $a_{0}, a_{1}, \dots, a_{n - 1}$ such that $M_{i,j} = a_{(j - i) \bmod n}$.
    \begin{equation}
        M = \begin{pmatrix}
            a_{0} & a_{1} & \dots & a_{n-1} \\
            a_{n-1} & a_{0} & \dots & a_{n-2} \\
            \vdots & \vdots & \ddots & \vdots \\
            a_{1} & a_{2} & \dots & a_{0}
        \end{pmatrix}
    \end{equation}
    Constructions, like circulant, are often adopted to reduce the search space and increase the probability of finding an MDS matrix~\cite{daemenBlockCipherSquare1997}.

    \section{Previous Work}

    In this section, we introduce some existing heuristic methods for synthesizing general linear transformations.

    Xiang~\etal~\cite{ToSC:XZLBZ20} gave good implementations of many matrices in terms of the gate count.
    They randomly picked one row or column operation that could reduce most number of ones in the matrix.
    If no operation could reduce the number of ones, they fall back to Gaussian Elimination.
    They also introduced several peephole optimization rules to further minimize the gate count.

    De~Brugi{\`e}re~\etal~\cite{brugiereReducingDepthLinear2021} proposed a cost-minimization algorithm to find low-depth CNOT circuits for reversible linear functions.
    Their approach maintains two sets: $L_{r}$ which contains row operations that can be applied in parallel, and $L_{c}$, which contains column operations that can be applied in parallel.
    The algorithm proceeds by randomly selecting an operation that can be inserted into $L_{c}$ or $L_{r}$ and minimize the cost function.
    If no such operation exists, the algorithms resets $L_{c}$ and $L_{r}$ to empty and start searching operations for next layer.
    The algorithm terminates when the current matrix becomes a permutation matrix, or when the depth exceeds a specific threshold.
    In their work, the authors employed the following four cost functions:
    \begin{gather}
        h_{sum}(A) = \sum_{i, j} a_{ij} , \\
        H_{sum}(A) = h_{sum}(A) + h_{sum}(A^{-1}) , \\
        h_{prod}(A) = \sum_{i} \log_{2} \left( \sum_{j} a_{ij} \right) , \\
        H_{prod}(A) = h_{prod}(A) + h_{prod}(A^{-1}) .
    \end{gather}
    Through experiments with random matrices, they observed that this depth-oriented greedy method performs well for small $n$ (typically $n < 40$).

    Zhu and Huang~\cite{zhuOptimizingDepthQuantum2023} introduced an algorithm called \texttt{one-way-opt} to optimize the depth of a given CNOT circuit.
    Reordering gate is a common approach to reduce the depth of quantum circuit.
    They reported that most quantum resource estimators only consider move-equivalence.
    That is, $CNOT(i, j)$ can be moved in front of $CNOT(i', j')$ if $\{ i, j \} \cap \{ i', j' \} = \varnothing$.
    They proposed that move-equivalence shold be relaexed to exchange-equivalence.
    That is, $CNOT(i, j)$ can be moved forward if $i \neq j'$ and $i' \neq j$.
    The \texttt{one-way-opt} algorithm maintains a set $L$ of CNOT gates that can be applied in parallel and tries to move gates into $L$ under exchange-equivalence.
    The algorithm is applied twice in forward and backward direction to reduce the depth of circuits given by Xiang~\etal~\cite{ToSC:XZLBZ20}.

    Yuan~\etal~improved upon Xiang~\etal's algorithm by extending reduction rules.
    They searched for complete reduction rules within length $K$ by converting it to a graph isomorphism problem.
    Depth reduction was achieved for many matrices compared with \cite{zhuOptimizingDepthQuantum2023}.

    Shi and Feng~\cite{AC:ShiFen24} gave state-of-the-art implementations of many matrices to our knowledge in terms of the circuit depth.
    They followed the framework of de Brugi{\`e}re~\etal~\cite{brugiereReducingDepthLinear2021} and used different cost functions:
    \begin{gather}
        \hat{H}_{prod}(A) = \max \{ h_{prod}(A) + h_{prod}((A^{-1})^{T}), h_{prod}(A^{T}) + h_{prod}(A^{-1}) \} , \\
        \hat{H}_{sq}(A) = \max \{ h_{sq}(A) + h_{sq}((A^{-1})^{T}), h_{sq}(A^{T}) + h_{sq}(A^{-1}) \} ,
    \end{gather}
    where
    \begin{equation}
        h_{sq}(A) = \sum_{i} \left( \sum_{j} a_{ij} \right)^{2} .
    \end{equation}
    They also checked whether a matrix can be transformed into permutation matrix using only 1 layer of CNOT gates before searching a new CNOT layer.

    Zhang~\etal~\cite{zhangOptimizedQuantumCircuit2024} proposed a new strategy for synthesizing AES MixColumn.
    They first viewed it as a matrix over $\mathbb{F}_{2^{8}}$ and transformed it into a simpler form.
    Later, they viewed it as a matrix over $\mathbb{F}_{2}$ and reduced it into identity matrix.
    They obtained the best known quantum circuit depth for AES MixColumn through this strategy.
    However, the proposed strategy has limitations worth noting.
    The transformation process relies on manual intervention, particularly in determining what constitutes a ``simple form'' -- a concept that remains undefined in their work.
    The second step requires careful adjustment of gate order, making the process difficult to automate \footnote{In the appendix, we provide an algorithm that generates a circuit distinct from Zhang~\etal's manual implementation, yet achieves the the same depth and size.}.

    \section{Optimizing the Synthesis of Circulant Matrices}

    In this section, we focus on circulant matrices.
    In Section 4.1, we introduce the method of \cite{zhangOptimizedQuantumCircuit2024}, which inspires us.
    In Section 4.2, we prove that the idea to simplify AES MixColumn as matrix over $\mathbb{F}_{2^8}$ can be generalized to transform circulant matrix over $\mathbb{F}_{2}[x]$ to upper-triangular matrix.
    In Section 4.3, we state that it is beneficial and usually affordable to explore different transformations.
    In Section 4.4, we give the observation that the recursion could stop early and the intuition behind this idea.
    In Section 4.5, we share some numerical results.

    \subsection{Method of \cite{zhangOptimizedQuantumCircuit2024} for AES MixColumn}

    As our method is inspired by \cite{zhangOptimizedQuantumCircuit2024}, we will describe their method in more details.
    Their method is targeting AES MixColumn only.
    The key point for the method of \cite{zhangOptimizedQuantumCircuit2024} for synthesizing AES MixColumn is that you should view it as a $4 \times 4$ matrix over $\mathbb{F}_{2^{8}}$ first.
    Previous methods, like \cite{AC:ShiFen24}, directly view it as a $32 \times 32$ matrix over $\mathbb{F}_{2}$.
    When it is treated as a matrix over $\mathbb{F}_{2^8}$, it is equivalent to considering the overall transformations between block matrices when viewed as a matrix over $\mathbb{F}_{2}$.
    This could enhance the circuit's parallelism and reduce the circuit depth.
    When viewed as a matrix over $\mathbb{F}_{2}$, heuristic methods, such as \cite{AC:ShiFen24}, are unlikely to capture the algebraic structure.

    The AES MixColumn is defined as
    \begin{equation}
        MC = \begin{pmatrix}
            02 & 03 & 01 & 01 \\
            01 & 02 & 03 & 01 \\
            01 & 01 & 02 & 03 \\
            03 & 01 & 01 & 02
        \end{pmatrix}_{\mathbb{F}_{2^8}} .
    \end{equation}
    Through search, they found that $MC$ can be transformed into a ``simpler matrix''
    \begin{equation}
        M_{1} = \begin{pmatrix}
            01 & 02 & 00 & 01 \\
            00 & 01 & 02 & 00 \\
            00 & 00 & 01 & 02 \\
            00 & 00 & 00 & 01
        \end{pmatrix}
    \end{equation}
    using elementary row and column operations.
    They argued that one should not use an even simpler matrix
    \begin{equation}
        M_{2} = \begin{pmatrix}
            01 & 02 & 00 & 00 \\
            00 & 01 & 02 & 00 \\
            00 & 00 & 01 & 02 \\
            00 & 00 & 00 & 01
        \end{pmatrix}
    \end{equation}
    which can also be obtained using elementary row and column operations.
    The reason is that using $M_{2}$ instead of $M_{1}$ needs one more layer of circuits in this step and the $01$ in the upper right corner of $M_{1}$ can be eliminated in parallel with subsequent operations.

    Later, they viewed $M_{1}$ as matrix over $\mathbb{F}_{2}$.
    Then, they designed a circuit of depth $4$ to nullify one $02$ in $M_{1}$.
    A circuit to transform $M_{1}$ to identity matrix can be obtained by putting the circuit to nullify $01$ and three circuits to nullify $02$ together.
    Through careful adjustment of gate order, they increased the parallelism of the circuit and reduced the depth of this part to $6$.
    Finally, they got a CNOT circuit for AES MixColumn of depth 10.

    The limitation of this method is that it requires human judgment on what is ``simpler matrix''.
    Later, we will generalize the idea of~\cite{zhangOptimizedQuantumCircuit2024} for AES MixColumn to circulant matrix and propose an algorithm without need of human intervention.

    \subsection{Our Method}

    In this subsection, we generalize the idea of \cite{zhangOptimizedQuantumCircuit2024} to circulant matrix over $\mathbb{F}_{2}[x]$.
    As $\mathbb{F}_{2^{k}}$ is obtained through $\mathbb{F}_{2}[x] / \langle g(x) \rangle$ where $g(x)$ is an irreducible polynomial, our method can cover circulant matrix over $\mathbb{F}_{2^{k}}$ for any $k$.
    Our method can also handle circulant matrix over $GL(\mathbb{F}_{2}, m)$ as there is a mapping from $GL(\mathbb{F}_{2}, m)$ to $\mathbb{F}_{2}[x]$, which preserves the property of addition.
    For example,
    \begin{equation}
        \begin{pmatrix}
            1 & 1 \\
            0 & 1
        \end{pmatrix} \mapsto (1, 1, 0, 1) \mapsto x^3 + x^2 + 0 \cdot x + 1 .
    \end{equation}

    The following theorem proves that circulant matrix over $\mathbb{F}_{2}[x]$ can be transformed into upper-triangular matrix.
    We replace the vague concept of ``simpler matrix'' in \cite{zhangOptimizedQuantumCircuit2024} with explicit upper-triangular matrix.
    The transformation has the chance to reduce the number of nonzero elements, as roughly half of the matrix are guaranteed to be zero.
    \begin{theorem} \label{thm:circ_tri}
        $2^{k} \times 2^{k}$ circulant matrix over $\mathbb{F}_{2}[x]$ can be transformed into upper-triangular matrix using $k 2^{k}$ elementary operations.
    \end{theorem}

    \begin{proof}
        Assume $M$ is a $2^{k} \times 2^{k}$ circulant matrix over $\mathbb{F}_{2}[x]$. As $M$ is circulant, it can be written as
        \begin{equation}
            M = \begin{pmatrix}
                A & B \\
                B & A
            \end{pmatrix}
        \end{equation}
        where $A$ and $B$ are $2^{k - 1} \times 2^{k - 1}$ matrices over $\mathbb{F}_{2}[x]$.
        Through the following operations, we can transform $M$ into a block upper-triangular matrix $M'$:
        \begin{equation} \label{eq:trans_circ_tri}
            \begin{pmatrix}
                A & B \\
                B & A
            \end{pmatrix}
            \xrightarrow{r_{2} \gets r_{2} + r_{1}}
            \begin{pmatrix}
                A & B \\
                A+B & A+B
            \end{pmatrix}
            \xrightarrow{c_{1} \gets c_{1} + c_{2}}
            \begin{pmatrix}
                A+B & B \\
                & A+B
            \end{pmatrix} = M' .
        \end{equation}
        If $A$ (or $B$) is $1 \times 1$ matrices over $\mathbb{F}_{2}[x]$, we have reached the goal.
        Otherwise, we claim that $A + B$ is a circulant matrix.
        Under this claim, we could apply the above sequence of operations recursively to $A + B$ and obtain a upper-triangular matrix.

        Now, we prove the claim. Assume
        \begin{equation}
            A = \begin{pmatrix}
                A_{1} & A_{2} \\
                A_{3} & A_{4}
            \end{pmatrix} ,
            B = \begin{pmatrix}
                B_{1} & B_{2} \\
                B_{3} & B_{4}
            \end{pmatrix} .
        \end{equation}
        As $M$ is circulant,
        \begin{equation}
            M = \begin{pNiceArray}{cc|cc}
                A_{1} & A_{2} & B_{1} & B_{2} \\
                A_{3} & A_{4} & B_{3} & B_{4} \\ \hline
                B_{1} & B_{2} & A_{1} & A_{2} \\
                B_{3} & B_{4} & A_{3} & A_{4}
            \end{pNiceArray}
        \end{equation}
        we have
        \begin{equation}
            A_{3} = B_{2}, A_{4} = A_{1}, B_{3} = A_{2}, B_{4} = B_{1} .
        \end{equation}
        Therefore,
        \begin{equation}
            A + B = \begin{pmatrix}
                A_{1} + B_{1} & A_{2} + B_{2} \\
                A_{3} + B_{3} & A_{4} + B_{4}
            \end{pmatrix}
            = \begin{pmatrix}
                A_{1} + B_{1} & A_{2} + B_{2} \\
                B_{2} + A_{2} & A_{1} + B_{1}
            \end{pmatrix} .
        \end{equation}
        To prove $A + B$ is circulant, we need to show that
        \begin{equation}
            \forall i, \forall j, (A + B)_{i, j} = (A + B)_{(i - 1) \bmod 2^{k - 1}, (j - 1) \bmod 2^{k - 1}} .
        \end{equation}
        Within $A_{1}$ (or equivalently $A_{2}$, $B_{1}$, $B_{2}$), the condition is naturally satisfied as it is a submatrix of $M$.
        For the left boundaries of and $A_{2} + B_{2}$, we have
        \begin{equation}
            \begin{aligned}
                (A_{2} + B_{2})_{0, 0} &= (A_{2})_{0, 0} + (B_{2})_{0, 0} \\
                &= (B_{3})_{2^{k-1} - 1, 2^{k-1} - 1} + (A_{3})_{2^{k-1} - 1, 2^{k-1} - 1} \\
                &= (A_{2})_{2^{k-1} - 1, 2^{k-1} - 1} + (B_{2})_{2^{k-1} - 1, 2^{k-1} - 1} \\
                &= (B_{2} + A_{2})_{2^{k-1} - 1, 2^{k-1} - 1}
            \end{aligned}
        \end{equation}
        and
        \begin{equation}
            \begin{aligned}
                (A_{2} + B_{2})_{0, j} &= (A_{2})_{0, j} + (B_{2})_{0, j} \\
                &= (A_{1})_{2^{k-1} - 1, j - 1} + (B_{1})_{2^{k-1} - 1, j - 1} \\
                &= (A_{1} + B_{1})_{2^{k-1} - 1, j - 1} .
            \end{aligned} \quad (j > 0)
        \end{equation}
        Similarly, readers can verify that the condition holds for other boundaries.

        Finally, we count the number of elementary operations.
        The recursion has depth $\log_{2}(2^{k}) = k$.
        In each layer, transformation \eqref{eq:trans_circ_tri} requires $2 \times \frac{1}{2} 2^{k} = 2^{k}$ elementary operations.
    \end{proof}

    \subsection{Improvement 1}

    In this subsection, we show that different choices of the row and column operations can lead to different results. Therefore, when affordable, it should be beneficial to try various transformations.

    Instead of transformation \eqref{eq:trans_circ_tri}, one could also use
    \begin{equation}
        \begin{pmatrix}
            A & B \\
            B & A
        \end{pmatrix}
        \xrightarrow{r_{1} \gets r_{1} + r_{2}}
        \begin{pmatrix}
            A+B & A+B \\
            B & A
        \end{pmatrix}
        \xrightarrow{c_{1} \gets c_{1} + c_{2}}
        \begin{pmatrix}
            & A+B \\
            A+B & A
        \end{pmatrix} = M'' .
    \end{equation}
    This leads to a different result $M''$, which may be better than $M'$.
    Even the following transformation may give a different result:
    \begin{equation} \label{eq:trans_circ_tri_2}
        \begin{pmatrix}
            A & B \\
            B & A
        \end{pmatrix}
        \xrightarrow{r_{1} \gets r_{1} + r_{2}}
        \begin{pmatrix}
            A+B & A+B \\
            B & A
        \end{pmatrix}
        \xrightarrow{c_{2} \gets c_{2} + c_{1}}
        \begin{pmatrix}
            A+B & \\
            B & A+B
        \end{pmatrix} = M''' .
    \end{equation}
    $M'''$ is the same as $M'$ under permutation.
    However, when recursively handling $A + B$ in $M'$ or $M'''$, the elementary operation chosen for $A + B$ may have different effect on $B$, which is not circulant in general.
    \begin{example}
        \begin{gather}
            \begin{aligned}
                M' &= \begin{pmatrix}
                    C_{1} & C_{2} & B_{1} & B_{2} \\
                    C_{2} & C_{1} & B_{3} & B_{4} \\
                    & & C_{1} & C_{2} \\
                    & & C_{2} & C_{1}
                \end{pmatrix} \\
                & \xrightarrow[c_{1} \gets c_{1} + c_{2}, c_{4} \gets c_{4} + c_{3}]{r_{1} \gets r_{1} + r_{2}, r_{4} \gets r_{4} + r_{3}}
                \begin{pmatrix}
                    & C_{1} + C_{2} & B_{1} + B_{2} + B_{3} + B_{4} & B_{2} + B_{3} \\
                    C_{1} + C_{2} & C_{1} & B_{3} + B_{4} & B_{4} \\
                    & & C_{1} & C_{1} + C_{2} \\
                    & & C_{1} + C_{2} &
                \end{pmatrix} ,
            \end{aligned} \\
            \begin{aligned}
                M''' &= \begin{pmatrix}
                    C_{1} & C_{2} & & \\
                    C_{2} & C_{1} & & \\
                    B_{1} & B_{2} & C_{1} & C_{2} \\
                    B_{3} & B_{4} & C_{2} & C_{1}
                \end{pmatrix} \\
                & \xrightarrow[c_{1} \gets c_{1} + c_{2}, c_{4} \gets c_{4} + c_{3}]{r_{1} \gets r_{1} + r_{2}, r_{4} \gets r_{4} + r_{3}}
                \begin{pmatrix}
                    & C_{1} + C_{2} & & \\
                    C_{1} + C_{2} & C_{1} & & \\
                    B_{1} + B_{2} & B_{2} & C_{1} & C_{1} + C_{2} \\
                    B_{1} + B_{2} + B_{3} + B_{4} & B_{2} + B_{4}  & C_{1} + C_{2} &
                \end{pmatrix}
            \end{aligned}
        \end{gather}
        where $C_{1} = A_{1} + B_{1}$, $C_{2} = A_{2} + B_{2}$.
        The same operation transformed $B$ into
        \begin{gather}
            B' = \begin{pmatrix}
                B_{1} + B_{2} + B_{3} + B_{4} & B_{2} + B_{3} \\
                B_{3} + B_{4} & B_{4}
            \end{pmatrix} , \\
            B''' = \begin{pmatrix}
                B_{1} + B_{2} & B_{2} \\
                B_{1} + B_{2} + B_{3} + B_{4} & B_{2} + B_{4}
            \end{pmatrix}
        \end{gather}
        respectively.
        They are not equivalent up to permutation.
    \end{example}

    For each $M$, there are $2$ possibilities for row operations and $2$ possibilities for column operations.
    This makes up $4$ different transformations.
    The order of row operation and column operation does not matter here.
    As matrix multiplication is associative, we can do the left-multiplication or the right-multiplication first.
    In recursion depth $d$, there are $2^{d - 1}$ such matrices and hence $4^{2^{d - 1}}$ different transformations.
    There are
    \begin{equation}
        \prod_{d = 1}^{k} 4^{2^{d - 1}} = 4^{\sum_{i = 0}^{k - 1} 2^{i}} = 4^{2^{k} - 1}
    \end{equation}
    cases in total.
    In the first layer, transformation \eqref{eq:trans_circ_tri} and \eqref{eq:trans_circ_tri_2} are equivalent up to permutation, which could save a $1/2$ factor for us.
    When $k = 1, 2, 3$, $\frac{1}{2} \cdot 4^{2^{k} - 1}$ equals $2, 32, 32768$ respectively.
    It is acceptable to try all possibilities.
    If $k \geq 4$, randomness or heuristic strategy may be introduced to handle more than half a billion possibilities.
    Luckily, $k \leq 3$ usually holds for cryptographic purposes.
    ($k = 3$ means the matrix is designed to be a $8 \times 8$ matrix over $\mathbb{F}_{2^{m}}$ for some $m$.)

    The pseudo-code description is given in Algorithm~\ref{alg:circ_tri}.
    \texttt{block\_size} described the size of matrix when the recursion has to stop.
    For example, \texttt{block\_size} is $8$ for AES MixColumn as it is a circulant matrix over $\mathbb{F}_{2^8}$ (or $\mathbb{F}_{2}^{8 \times 8}$).
    \texttt{step} is used to implement recursion and should be $2^{n-1}$ at the beginning.
    The return value is a list of transformed matrices $M''$s and corresponding row operations $C''_{r}$ and column operations $C''_{c}$.

    \begin{algorithm}
        \caption{Algorithm to convert circulant matrix into matrices that may have low-depth CNOT circuit}
        \label{alg:circ_tri}
        \begin{algorithmic}
            \Require $2^{n} \times 2^{n}$ matrix $M$ over $\mathbb{F}_{2}$, which is circulant over $GL(\mathbb{F}_{2}, block\_size)$
            \Ensure List of candidate matrices and corresponding operations
            \Procedure{ProcessCirculant}{$M, block\_size, step$}
            \If{$step < block\_size$}
            \State \Return $\{ (M, (\varnothing, \varnothing)) \}$
            \EndIf
            \State $L \gets \varnothing$
            \If{$2^{n} / block\_size \leq \texttt{THRESHOLD}$}
            \Comment{typically \texttt{THRESHOLD} = $2^{1},2^{2},2^{3}$}
            \LComment{If $block\_size = 2^{n-1}$, could return just 00 and 01}
            \State $S \gets $ all 01-strings of length $2^{n} / step$
            \Else
            \State $S \gets $ a 01-string of length $2^{n} / step$
            \Comment{randomly or heuristically}
            \EndIf
            \For{$str \in S$}
            \State $M' \gets M$
            \State $C'_{r} \gets \varnothing$, $C'_{c} \gets \varnothing$
            \For{$i = 1, \dots, 2^{n-1} / step$}
            \Comment{row operations}
            \If{next bit in $str$ is $0$}
            \For{$j = 1, \dots, step$}
            \State $r_{1} \gets (i - 1) \times (2 \times step) + j$
            \State $r_{2} \gets r_{1} + step$
            \State Add row $r_{2}$ to row $r_{1}$ of $M'$
            \State $C'_{r}$.append($r_{2}, r_{1}$)
            \EndFor
            \Else
            \For{$j = 1, \dots, step$}
            \Comment{add row $r_{1}$ to row $r_{2}$}
            \State \dots
            \EndFor
            \EndIf
            \EndFor
            \For{$i = 1, \dots, 2^{n-1} / step$}
            \Comment{column operations}
            \State \dots
            \EndFor
            \For{$(M'', (C''_{r}, C''_{c})) \in \Call{ProcessCirculant}{M', block\_size, step / 2}$}
            \State $C''_{r}$.extend($C'_{r}$), $C''_{c}$.extend($C'_{c}$)
            \State $L$.append($M'', (C''_{r}, C''_{c})$)
            \EndFor
            \EndFor
            \State \Return $L$
            \EndProcedure
        \end{algorithmic}
    \end{algorithm}

    The algorithm first checks whether the recursion should stop.
    If so, just return current matrix and no operation is needed.
    $S$ is a collection of 01-strings that guides which operation should be performed.
    We checks the maximum recursion depth (i.e. $\log(2^{n} / \texttt{block\_size})$).
    If it is small enough, we try all possible transformations.
    Otherwise, we use some strategy, such as randomness, to determine the transformation.
    According to $S$, we produce transformed $M'$ and record corresponding $C_{r}'$ and $C_{c}'$.
    Then we recursively handle $M'$ by halving \texttt{step}.
    We concatenate $C_{r}'$ and returned $C_{r}''$ to get the row operation to transform $M$ into $M''$.

    \begin{example}
        Let $M = \begin{pmatrix}
            A & B \\
            B & A
        \end{pmatrix}$ where $A$ and $B$ is $2^{b} \times 2^{b}$ matrix.
        The algorithm should return
        \begin{gather}
            \begin{pmatrix}
                & A + B \\
                A + B & A
            \end{pmatrix}, \{ r_{1} \gets r_{1} + r_{2^{b} + 1}, \cdots \}, \{ c_{1} \gets c_{1} + c_{2^{b} + 1}, \dots \} , \\
            \begin{pmatrix}
                A + B & \\
                B & A + B
            \end{pmatrix}, \{ r_{1} \gets r_{1} + r_{2^{b} + 1}, \cdots \}, \{ c_{2^{b} + 1} \gets c_{2^{b} + 1} + c_{1}, \dots \} .
        \end{gather}
    \end{example}

    \subsection{Improvement 2}

    In this subsection, we share the idea of stopping recursion early. When the recursion depth is larger, the overhead might grow faster than the benefit.

    The second modification comes from the following intuition:
    in recursion depth 1, it guarantees $1/4$ of the matrix is $0$;
    however, in recursion depth 2, the same cost only guarantees $1/8$ of the matrix is $0$.
    The gain seems to decrease exponentially as the depth grows. (This is not technically correct, but serves as a good intuition.)
    Therefore, we could stop the recursion early.
    Experiments on matrices collected by Kranz~\etal~\cite{ToSC:KLSW17} show that stopping recursion at depth $3$, if ever possible, is almost unworthy compared to depth $1$ or $2$.
    Besides, for matrices that can be implemented with extremely low quantum circuit depth like PRIDE~\cite{C:ADKLPY14}, we should fallback to call underlying optimizer for general CNOT circuits directly.

    The detailed algorithm is given in Algorithm~\ref{alg:circ_full} with procedure \textsc{ProcessCirculant} in Algorithm~\ref{alg:circ_tri}.
    Every possible recursion depth is explored through changing \texttt{block\_size}.
    By increasing \texttt{block\_size}, we reduce the recursion depth, which equals $\log(2^{n} / \texttt{block\_size})$.
    We view the matrix $M'$ returned by \textsc{ProcessCirculant} as a matrix over $\mathbb{F}_{2}$ and call heuristic algorithm \textsc{Heu} to implement $M'$.
    Note that \textsc{Heu} is allowed to implement $\pi(M')$ and return the permutation $\pi$.
    The implementation $C$ of $M'$ together with $C_{r}$ and $C_{c}$ given by \textsc{ProcessCirculant} forms what we need in Eq~\eqref{eq:solution}.
    If we allow \texttt{block\_size} to be $2^{n}$, \textsc{ProcessCirculant} will simply return $M$.
    This is equivalent to call \textsc{Heu} directly.
    We call this as fallback strategy.

    \begin{algorithm}
        \caption{Algorithm for synthesizing generalized circulant matrix}
        \label{alg:circ_full}
        \begin{algorithmic}
            \Require $2^{n} \times 2^{n}$ generalized circulant matrix $M$ over $\mathbb{F}_{2}$
            \Ensure CNOT circuit for $M$
            \State $b \gets \min \{ k \geq 0: M \text{ is circulant matrix over } GL(\mathbb{F}_{2}, 2^{k}) \}$
            \State $bestC \gets \perp$
            \For{$block\_size \in \{ 2^{b}, 2^{b+1}, \dots, 2^{n-1}, 2^{n} \}$}
            \Comment{fallback: allow $block\_size = 2^{n}$}
            \State $L \gets \Call{ProcessCirculant}{M, block\_size, 2^{n-1}}$
            \Comment{no-op when $block\_size = 2^{n}$}
            \For{$(M', (C_{r}, C_{c})) \in L$}
            \LComment{\textsc{Heu}: heuristic algorithm for synthesizing CNOT circuits}
            \State $C, \pi \gets \Call{Heu}{M'}$
            \Comment{$C$ implements $\pi(M')$}
            \State $C_{c} \gets \Call{transpose}{C_{c}}$
            \Comment{change $c_{1} \gets c_{1} + c_{2}$ to $CNOT(c_{1},c_{2})$}
            \State $C_{r} \gets \Call{reverse}{C_{r}}$
            \State $C_{r} \gets \pi(C_{r})$
            \State $C \gets \Call{concat}{C_{c}, C, C_{r}}$
            \Comment{$C$ implements $\pi(M)$}
            \If{$C$ is better than $bestC$}
            \State $C \gets bestC$
            \EndIf
            \EndFor
            \EndFor
            \State \Return $bestC$
        \end{algorithmic}
    \end{algorithm}

    Recall that by left-multiplying $E(i + j)$ to $M$, we add the $j$-th row to $i$-th row of $M$.
    By right-multiplying $E(i + j)$ to $M$, we add $i$-th column to $j$-th column of $M$.
    Hence, $(j, i)$ in $C_{r}$ represents $E(i + j)$, whereas $(j, i)$ in $C_{c}$ represents $E(j + i)$.
    We change $(j, i)$ in $C_{c}$ to $(i, j)$ instead.
    After that, $(j', i')$ in both $C_{r}$ and $C_{c}$ means $E(i' + j')$ or $CNOT(j',i')$.

    To get the correct circuit, we reverse the order of operations in $C_{r}$ (as \textsc{ProcessCirculant} collects it in order $\{ E(i_{s} + j_{s}), \dots, E(i_{1} + j_{1}) \}$).
    By Eq~\eqref{eq:type-3-move}, we can absorb $\pi$ into $C_{r}$ to get a circuit implementing $\pi(M)$.

    \begin{example}
        Let
        \begin{equation}
            M = \begin{pNiceArray}{cc|cc}
                1 & 1 & & 1 \\
                & 1 & & \\ \hline
                & 1 & 1 & 1 \\
                & & & 1
            \end{pNiceArray} .
        \end{equation}
        Assume $C_{r}$ is $\{ r_{1} \gets r_{1} + r_{3}, r_{2} \gets r_{2} + r_{4} \}$ and $C_{c}$ is $\{ c_{1} \gets c_{1} + c_{3}, c_{2} \gets c_{2} + c_{4} \}$.
        Then,
        \begin{equation}
            M' = \begin{pmatrix}
                & & 1 & \\
                & & & 1 \\
                1 & & 1 & 1 \\
                & 1 & & 1
            \end{pmatrix} =
            \begin{pmatrix}
                & & 1 & \\
                & & & 1 \\
                1 & & & \\
                & 1 & &
            \end{pmatrix} E(1 + 3) E(1 + 4) E(2 + 4) .
        \end{equation}
        We can check that
        \begin{equation}
            \begin{pmatrix}
                & & 1 & \\
                & & & 1 \\
                1 & & & \\
                & 1 & &
            \end{pmatrix} \underbrace{E(\pi(2) + \pi(4)) E(\pi(1) + \pi(3))}_{\text{row operation}} \cdot E(1 + 3) E(1 + 4) E(2 + 4) \cdot \underbrace{E(3 + 1) E(4 + 2)}_{\text{column operation}} =  M
        \end{equation}
        where $\pi(1) = 3, \pi(2) = 4, \pi(3) = 1, \pi(4) = 2$.
    \end{example}

    \subsection{Numerical Results}

    \cite{ToSC:KLSW17} collected a set of linear layers from block ciphers and constructed MDS matrices, which has been adopted by studies, such as \cite{ToSC:XZLBZ20,zhuOptimizingDepthQuantum2023,AC:ShiFen24,ToSC:YWSZZ24}, to benchmark the performance for synthesizing linear layers.
    The data set is available at \url{https://github.com/rub-hgi/shorter_linear_slps_for_mds_matrices}.
    In this subsection, we will test our algorithm on this data set.
    Note that it contains some matrices that are not circulant, which will be filtered out.
    We will disable fallback strategy to better reflect the effect of our method.

    \subsubsection{Quantum Circuit Depth}

    The DW cost is widely used to estimate the quantum threat to existing symmetric encryption scheme.
    Reducing the depth of quantum circuit for linear layers can give a more precision estimation.
    Besides, as the decoherence time of qubits is limited, quantum circuit with larger depth is likely to be error prone.

    The greedy algorithm of~\cite{AC:ShiFen24} can achieve a good depth for many matrices collected by~\cite{ToSC:KLSW17}.
    We use it as the base for the \textsc{Heu} required by Algorithm~\ref{alg:circ_full}.
    We made some modification to boost the performance.
    Besides $\hat{H}_{prod}$ and $\hat{H}_{sq}$, we also use $\hat{H}_{sum} = H_{sum}$ as cost function, which ourperforms $\hat{H}_{prod}$ or $\hat{H}_{sq}$ in some cases.
    We run Algorithm~\ref{alg:circ_full} multiple times to obtain the results.
    The results are shown in \tablename~\ref{tbl:result_cipher} and \tablename~\ref{tbl:result_mds}.

    \begin{table}[htp]
        \centering
        \caption{Comparison of the depth/gate count of CNOT circuits for matrices used in block ciphers}
        \label{tbl:result_cipher}
        \begin{threeparttable}
            \begin{tabular}{c|c|cc|cc|c}
                \hline
                Cipher & Size & \cite{zhuOptimizingDepthQuantum2023} & \cite{ToSC:YWSZZ24} & \cite{brugiereReducingDepthLinear2021}\tnote{2} & \cite{AC:ShiFen24} & \textbf{Ours}\tnote{1} \\
                \hline
                AES~\cite{daemenDesignRijndaelAdvanced2020}\tnote{3} & 32
                & 28/92 & 13/98 & 12/128 & 10/131\tnote{*} & \textbf{10/107}\tnote{*,4,ii} \\
                Anubis~\cite{barretop.s.l.m.AnubisBlockCipher2000} & 32
                & 20/98 & 14/102 & 14/136 & 10/119\tnote{*} & \textbf{10/117}\tnote{*,i} \\
                Clefia $M_0$~\cite{FSE:SSAMI07} & 32
                & 27/98 & 15/105 & 13/126 & 10/110\tnote{*} & \textbf{10/120}\tnote{*,i} \\
                Clefia $M_1$~\cite{FSE:SSAMI07} & 32
                & 16/103 & 13/106 & 13/127 & 10/128\tnote{*} & \textbf{11/124}\tnote{i} \\
                Joltik~\cite{jeanJoltikV132015} & 16
                & 17/44 & 9/48 & 9/48 & 7/52\tnote{*} & \textbf{8/52}\tnote{i} \\
                SmallScale AES~\cite{FSE:CidMurRob05} & 16
                & 19/43 & 10/46 & 11/59 & 10/62 & \textbf{8/53}\tnote{*,i} \\
                Whirlwind $M_0$~\cite{DCC:BNNRT10} & 32
                & 51/183 & 46/249 & 28/286 & 28/331 & \textbf{17/200}\tnote{*,i} \\
                Whirlwind $M_1$~\cite{DCC:BNNRT10} & 32
                & 54/190 & 50/264 & 25/279 & 22/290 & \textbf{19/241}\tnote{*,i} \\
                \hline
                MIDORI~\cite{AC:BBISHA15} & 16
                & 3/24\tnote{*} & 4/24 & 3/24\tnote{*} & 3/24\tnote{*} & \textbf{3/24}\tnote{*} \\
                PRIDE $L_0$~\cite{C:ADKLPY14} & 16
                & 3/24\tnote{*} & - & 3/24\tnote{*} & 3/24\tnote{*} & \textbf{3/24}\tnote{*} \\
                PRIDE $L_3$~\cite{C:ADKLPY14} & 16
                & 5/24 & - & 3/24\tnote{*} & 3/24\tnote{*} & \textbf{4/24} \\
                PRINCE $M_0$~\cite{AC:BCGKKK12} & 16
                & 6/24 & - & 3/24\tnote{*} & 3/24\tnote{*} & \textbf{4/28} \\
                PRINCE $M_1$~\cite{AC:BCGKKK12} & 16
                & 6/24 & - & 3/24\tnote{*} & 3/24\tnote{*} & \textbf{4/28} \\
                QARMA64~\cite{ToSC:Avanzi17} & 16
                & 5/24 & - & 3/24\tnote{*} & 3/24\tnote{*} & \textbf{3/24}\tnote{*} \\
                QARMA128~\cite{ToSC:Avanzi17} & 32
                & 5/48 & 5/48 & 3/48\tnote{*} & 3/48\tnote{*} & \textbf{7/96} \\
                \hline
            \end{tabular}
            \begin{tablenotes}
                \item[*] Results with (*) indicates that they have the state-of-the-art depth.
                \item[i] Results with (i) indicates that they are derived with recursion depth 1.
                \item[ii] Results with (ii) indicates that they are derived with recursion depth 2.
                \item[1] Use \cite{AC:ShiFen24} as \textsc{Heu} without fallback.
                \item[2] Result is not available in original paper and is provided by \cite{AC:ShiFen24}.
                \item[3] \cite{zhangOptimizedQuantumCircuit2024} gave a result of 10/105, which is better than 10/107 in the table.
                \item[4] This can be improved to 10/105 using Algorithm~\ref{alg:upper_tri} in Appendix. which requires the matrix can be transformed into unit upper-triangular form.
            \end{tablenotes}
        \end{threeparttable}
    \end{table}

    \begin{table}[htp]
        \centering
        \caption{Comparison of the depth/gate count of CNOT circuits for constructed MDS matrices}
        \label{tbl:result_mds}
        \begin{threeparttable}
            \begin{tabular}{c|c|cc|cc|c}
                \hline
                Matrices & Size & \cite{zhuOptimizingDepthQuantum2023} & \cite{ToSC:YWSZZ24} & \cite{brugiereReducingDepthLinear2021} & \cite{AC:ShiFen24} & \textbf{Ours} \\
                \hline
                \multicolumn{7}{c}{$4 \times 4$ matrices over $GL(4, \mathbb{F}_{2})$} \\
                \hline
                BKL16~\cite{C:BeiKraLea16} & 16
                & 21/41 & 10/43\tnote{*} & 12/57 & 10/59\tnote{*} & \textbf{11/70}\tnote{i} \\
                LS16~\cite{FSE:LiuSim16} & 16
                & 26/44 & 12/49 & 12/65 & 11/63 & \textbf{10/68}\tnote{*,i} \\
                LW16~\cite{FSE:LiWan16} & 16
                & 27/44 & 12/48 & 12/65 & 11/62 & \textbf{10/65}\tnote{*,i} \\
                SKOP15~\cite{FSE:SKOP15} & 16
                & 22/44 & 11/47 & 11/59 & 11/59 & \textbf{8/58}\tnote{*,i} \\
                LW16(Involutory)~\cite{FSE:LiWan16} & 16
                & 25/44 & - & 9/48 & 7/52\tnote{*} & \textbf{9/59}\tnote{i} \\
                SKOP15(Involutory)~\cite{FSE:SKOP15} & 16
                & 16/44 & - & 9/48 & 7/52\tnote{*} & \textbf{8/52}\tnote{i} \\
                \hline
                \multicolumn{7}{c}{$4 \times 4$ matrices over $GL(8, \mathbb{F}_{2})$} \\
                \hline
                BKL16~\cite{C:BeiKraLea16} & 32
                & 47/114 & 26/132 & 20/188 & 18/208\tnote{*} & \textbf{18/226}\tnote{*,ii} \\
                LS16~\cite{FSE:LiuSim16} & 32
                & 54/121 & 29/136 & 23/203 & 21/235 & \textbf{17/207}\tnote{*,ii} \\
                LW16~\cite{FSE:LiWan16} & 32
                & 42/104 & 27/129 & 16/167 & 13/164\tnote{*} & \textbf{16/185}\tnote{ii} \\
                SKOP15~\cite{FSE:SKOP15} & 32
                & 20/90 & 14/96 & 11/118 & 10/112 & \textbf{9/116}\tnote{*,i}  \\
                LW16(Involutory)~\cite{FSE:LiWan16} & 32
                & 19/87 & 11/96 & 8/98\tnote{*} & 8/99\tnote{*} & \textbf{8/110}\tnote{*,i} \\
                SKOP15(Involutory)~\cite{FSE:SKOP15} & 32
                & 16/91 & 13/94 & 9/96 & 8/101\tnote{*} & \textbf{8/104}\tnote{*,i} \\
                \hline
                \multicolumn{7}{c}{$8 \times 8$ matrices over $GL(4, \mathbb{F}_{2})$} \\
                \hline
                SKOP15~\cite{FSE:SKOP15} & 32
                & 49/170 & 48/249 & 29/286 & 28/349 & \textbf{20/230}\tnote{*,i} \\
                SKOP15(Involutory)~\cite{FSE:SKOP15} & 32
                & 37/185 & 46/248 & 30/300 & 29/337 & \textbf{23/261}\tnote{*,i} \\
                \hline
            \end{tabular}
        \end{threeparttable}
    \end{table}

    For various matrices, our algorithm achieves reduced depth.
    Notably, compared with previous works, we are able to reduce the depth by 39\% for Whirlwind $M_0$.
    As for AES MixColumn, our result is quite close to state-of-the-art given by Zhang~\etal~\cite{zhangOptimizedQuantumCircuit2024}.
    The comparison with previous works on AES MixColumn is shown in \tablename~\ref{tbl:result_aes_mc}.

    \begin{table}[htp]
        \centering
        \caption{Comparison of CNOT circuits of the AES MixColumns matrix}
        \label{tbl:result_aes_mc}
        \begin{threeparttable}
            \begin{tabular}{ccccc}
                \hline
                Source & Width & Depth & \#CNOT & Manual \\
                \hline
                \cite{zhuOptimizingDepthQuantum2023} & 32 & 28 & 92 & No \\
                \cite{AC:LPZW23} & 32 & 16 & 98 & No \\
                \cite{AC:ShiFen24} & 32 & 10 & 131 & No \\
                \cite{zhangOptimizedQuantumCircuit2024} & 32 & 10 & 105 & Yes \\
                \textbf{Ours}(Alg.~\ref{alg:circ_full}) & \textbf{32} & \textbf{10} & \textbf{107} & \textbf{No} \\
                \textbf{Ours}(Alg.~\ref{alg:upper_tri}) & \textbf{32} & \textbf{10} & \textbf{105} & \textbf{No}\tnote{1} \\
                \hline
            \end{tabular}
            \begin{tablenotes}
                \item[1] Algorithm~\ref{alg:upper_tri} is specifically designed for the subtask in \cite{zhangOptimizedQuantumCircuit2024} that requires manual adjustment.
            \end{tablenotes}
        \end{threeparttable}
    \end{table}

    The great performance for Whirlwind might be related to its high density, i.e., proportion of ones.
    When a matrix is dense and ones are evenly distributed, heuristic methods, such as~\cite{AC:ShiFen24}, are more likely to produce inefficient implementations.
    Our approach allows heuristic method to find more efficient implementations by reducing the number of ones.
    For Whirlwind and SKOP15 8x8 over $GL(4, \mathbb{F}_{2})$, the density is above 45\% and the depth reduction is relatively significant.
    For the rest matrices, the density is below 35\% and the improvement is milder or even negative.

    \subsubsection{Classical Circuit Size}

    The circuit size is a concern to lightweight cryptography, as smaller size means smaller area on chip.
    This is important when the device is resource constrained.
    Note that as different variants of algorithms are used, the data in this subsection (\tablename~\ref{table:result_size_cipher}) is irrelevant of that in previous subsection, like \tablename~\ref{tbl:result_cipher}.

    The algorithm~\cite{ToSC:XZLBZ20} is good at reducing the gate count.
    Its improvement~\cite{ToSC:YWSZZ24} has better performance.
    However, it is not open source and the improvement over~\cite{ToSC:XZLBZ20} is not significant, which should not affect whether our algorithm is effective.
    Hence, we use \cite{ToSC:XZLBZ20} as the \textsc{Heu} required by Algorithm~\ref{alg:circ_full}.
    The results are shown in \tablename~\ref{table:result_size_cipher}.

    \begin{table}[htp]
        \centering
        \caption{Comparison of the gate count of XOR circuits for matrices used in block ciphers}
        \label{table:result_size_cipher}
        \begin{threeparttable}
            \begin{tabular}{c|c|c|cc|c}
                \hline
                \multirow{2}{*}{Cipher} & \multirow{2}{*}{Size} & g-XOR & \multicolumn{3}{c}{s-XOR} \\
                \cline{3-6}
                & & \cite{linFrameworkOptimizeImplementations2021} & \cite{ToSC:XZLBZ20} & \cite{ToSC:YWSZZ24} & \textbf{Ours}\tnote{1} \\
                \hline
                AES~\cite{daemenDesignRijndaelAdvanced2020} & 32
                & 91\tnote{*} & 92 & 91\tnote{*} & \textbf{105} \\
                Anubis~\cite{barretop.s.l.m.AnubisBlockCipher2000} & 32
                & 96\tnote{*} & 99 & - & \textbf{116} \\
                Clefia $M_0$~\cite{FSE:SSAMI07} & 32
                & 96\tnote{*} & 98 & 97 & \textbf{115} \\
                Clefia $M_1$~\cite{FSE:SSAMI07} & 32
                & 108 & 103\tnote{*} & 103\tnote{*} & \textbf{117} \\
                Joltik~\cite{jeanJoltikV132015} & 16
                & 43\tnote{*} & 44 & 44 & \textbf{50} \\
                SmallScale AES~\cite{FSE:CidMurRob05} & 16
                & 43 & 43 & 42\tnote{*} & \textbf{50} \\
                Whirlwind $M_0$~\cite{DCC:BNNRT10} & 32
                & 183 & 183 & 173 & \textbf{159}\tnote{*} \\
                Whirlwind $M_1$~\cite{DCC:BNNRT10} & 32
                & 180 & 190 & 181 & \textbf{169}\tnote{*} \\
                \hline
            \end{tabular}
            \begin{tablenotes}
                \item[*] Results with (*) indicates that they have the state-of-the-art XOR count.
                \item[1] Use \cite{ToSC:XZLBZ20} as \textsc{Heu} without fallback.
            \end{tablenotes}
        \end{threeparttable}
    \end{table}

    It turns out that generally our method has a large overhead in terms of size.
    However, for some matrices, like Whirlwind $M_{0}$, we beat the previous state-of-the-art result.

    \section{Conclusion}

    In this paper, we present a new framework to optimize the linear layers used in symmetric cryptography.
    Our method makes use of the circulant structure of matrices, which is commonly adopted in the design of MDS matrices.
    To the best of our knowledge, this is the first time that such property is used to optimize the result.
    We achieve superior results in many matrices.
    For Whirlwind $M_{0}$, we reduced the depth of quantum circuit from 28 to 17.
    For AES MixColumn, our automated method achieves the state-of-the-art depth and only uses 2 extra gates than manually optimized result by Zhang~\etal~\cite{zhangOptimizedQuantumCircuit2024}.
    In terms of XOR count of classical circuit, we could also beat the previous state-of-the-art results for some specific matrices using this framework.
    The results show that the new approach is effective and has the potential to improve both quantum and classical circuits.

    In the future, it would be interesting to explore whether other properties used in the design of linear layers can be used to improve their circuit implementation.

    \bibliography{ref}

    \clearpage
    \appendix
    \section{Algorithm for Unit Upper-Triangular matrix}

    \cite{zhangOptimizedQuantumCircuit2024} observed that AES MixColumn can be transformed into upper triangular form through restricted elementary operations.
    \begin{equation}
        \begin{pmatrix}
            01 & 02 & 00 & 01 \\
            & 01 & 02 & 00 \\
            &    & 01 & 02 \\
            &    &    & 01
        \end{pmatrix}_{\mathbb{F}_{2^{8}}}
    \end{equation}
    where
    \begin{gather}
        (01)_{\mathbb{F}_{2^{8}}} = \begin{pmatrix}
            1 &   &   &   &   &   &   &   \\
            & 1 &   &   &   &   &   &   \\
            &   & 1 &   &   &   &   &   \\
            &   &   & 1 &   &   &   &   \\
            &   &   &   & 1 &   &   &   \\
            &   &   &   &   & 1 &   &   \\
            &   &   &   &   &   & 1 &   \\
            &   &   &   &   &   &   & 1
        \end{pmatrix}_{\mathbb{F}_{2}} , \\
        (02)_{\mathbb{F}_{2^{8}}} = \begin{pmatrix}
            &   &   &   &   &   &   & 1 \\
            1 &   &   &   &   &   &   & 1 \\
            & 1 &   &   &   &   &   &   \\
            &   & 1 &   &   &   &   & 1 \\
            &   &   & 1 &   &   &   & 1 \\
            &   &   &   & 1 &   &   &   \\
            &   &   &   &   & 1 &   &   \\
            &   &   &   &   &   & 1 &
        \end{pmatrix}_{\mathbb{F}_{2}} .
    \end{gather}

    We notice that it is actually a unit upper-triangular sparse matrix over $\mathbb{F}_{2}$.
    This structure enables efficient synthesis using the following algorithm.

    \begin{algorithm}
        \caption{Algorithm to synthesize unit upper-triangular matrix}
        \label{alg:upper_tri}
        \begin{algorithmic}
            \Require A $n \times n$ unit upper-triangular matrix $M$ over $\mathbb{F}_{2}$
            \Ensure CNOT circuit for $M$
            \State $C \gets \varnothing$
            \While{$M \neq I$}
            \State $S \gets \{ i: \forall k \neq i, M[i][k] = 0 \}$
            \State $T \gets \{ j: \exists i \in S, i \neq j, M[j][i] = 1 \}$
            \State $L \gets \varnothing$
            \For{$j \in T$}
            \For{$i \in S$}
            \If{$M[j][i] = 1$ and $i \notin L$ and $j \notin L$}
            \State Add $i$-th row to $j$-th row of $M$
            \State Add $CNOT(i, j)$ to $C$
            \State $L \gets L \cup \{ i, j \}$
            \EndIf
            \EndFor
            \EndFor
            \EndWhile
            \State \Return $C$
        \end{algorithmic}
    \end{algorithm}

    \clearpage
    \section{The Depth 17 Implementation of Whirlwind $M_0$}

    \begin{table}[H]
        \centering
        \caption{The implementation of Whirlwind $M_0$ with quantum depth 17. $(i, j)$ stands for the CNOT gate adding the $i$-th qubit to the $j$-th qubit. The outputs $\ket{y_{0}}, \ket{y_{1}}, \dots, \ket{y_{31}}$ are represented by 28, 27, 25, 18, 24, 31, 29, 22, 20, 19, 17, 26, 16, 23, 21, 30, 1, 2, 11, 4, 5, 6, 15, 0, 9, 10, 3, 12, 13, 14, 7, 8 respectively.}
        \begin{tabular}{|c|*{8}{c}|}
            \hline
            \multirow{2}{*}{Depth 1} & (0,16) & (1,17) & (2,18) & (3,19) & (4,20) & (5,21) & (6,22) & (7,23) \\
            & (8,24) & (9,25) & (10,26) & (11,27) & (12,28) & (13,29) & (14,30) & (15,31) \\
            \hline
            \multirow{2}{*}{Depth 2} & (10,4) & (6,8) & (14,0) & (2,12) & (22,24) & (30,16) & (18,28) & (26,20) \\
            & (3,13) & (7,9) & (11,5) & (15,1) & (31,17) & (19,29) & (23,25) & (27,21) \\
            \hline
            \multirow{2}{*}{Depth 3} & (25,22) & (29,18) & (21,26) & (17,30) & (1,14) & (13,2) & (5,10) & (9,6) \\
            & (8,11) & (12,15) & (4,7) & (0,3) & (24,27) & (16,19) & (28,31) & (20,23) \\
            \hline
            \multirow{2}{*}{Depth 4} & (28,25) & (24,29) & (16,21) & (20,17) & (3,11) & (7,15) & (30,1) & (22,9) \\
            & (26,5) & (18,13) & (23,10) & (31,2) & (27,6) & (19,14) & & \\
            \hline
            \multirow{2}{*}{Depth 5} & (23,16) & (19,20) & (31,24) & (27,28) & (4,12) & (0,8) & (17,18) & (21,22) \\
            & (15,6) & (11,2) & (1,3) & (25,26) & (13,7) & (29,30) & (9,10) & \\
            \hline
            \multirow{2}{*}{Depth 6} & (11,4) & (22,12) & (18,8) & (15,0) & (20,27) & (1,2) & (16,31) & (24,23) \\
            & (28,19) & (9,3) & (5,6) & (13,14) & (30,10) & (26,7) & & \\
            \hline
            \multirow{2}{*}{Depth 7} & (26,8) & (15,14) & (29,4) & (25,0) & (17,20) & (21,16) & (11,10) & (30,12) \\
            & (31,3) & (5,7) & & & & & & \\
            \hline
            Depth 8 & (21,12) & (17,8) & (4,3) & (20,7) & & & & \\
            \hline
            Depth 9 & (8,15) & (12,11) & & & & & & \\
            \hline
            Depth 10 & (29,12) & (25,8) & (30,3) & (0,7) & & & & \\
            \hline
            Depth 11 & (12,9) & (8,13) & (24,11) & (25,14) & (28,15) & (29,10) & (16,3) & (27,7) \\
            \hline
            \multirow{2}{*}{Depth 12} & (8,5) & (12,1) & (29,24) & (25,28) & (18,0) & (22,4) & (23,3) & (19,7) \\
            & (20,15) & (21,2) & (17,6) & (16,11) & (30,9) & & & \\
            \hline
            \multirow{2}{*}{Depth 13} & (7,15) & (3,11) & (16,29) & (20,25) & (28,17) & (24,21) & (4,12) & (0,8) \\
            & (26,13) & (30,10) & (18,5) & (22,1) & & & & \\
            \hline
            \multirow{2}{*}{Depth 14} & (22,8) & (30,0) & (18,12) & (26,4) & (7,13) & (3,9) & (15,5) & (11,1) \\
            & (31,24) & (27,28) & (19,20) & (23,16) & & & & \\
            \hline
            \multirow{2}{*}{Depth 15} & (14,8) & (6,0) & (10,12) & (2,4) & (24,13) & (28,9) & (20,1) & (16,5) \\
            & (22,31) & (26,19) & (30,23) & (18,27) & & & & \\
            \hline
            \multirow{2}{*}{Depth 16} & (17,11) & (21,15) & (25,3) & (29,7) & (19,2) & (31,14) & (23,6) & (27,10) \\
            & (18,20) & (22,16) & (26,28) & (30,24) & (8,5) & (12,1) & (4,9) & (0,13) \\
            \hline
            \multirow{2}{*}{Depth 17} & (1,28) & (2,27) & (11,25) & (4,18) & (5,24) & (6,31) & (15,29) & (0,22) \\
            & (9,20) & (10,19) & (3,17) & (12,26) & (13,16) & (14,23) & (7,21) & (8,30) \\
            \hline
        \end{tabular}
    \end{table}

    \clearpage
    \section{The 159 XOR Implementation of Whirlwind $M_0$}

    \begin{table}[H]
        \centering
        \caption{The implementation of Whirlwind $M_0$ with 159 XORs. $(i, j)$ stands for the XOR gate adding the $i$-th bit to the $j$-th bit. The outputs $y_{0}, y_{1}, \dots, y_{31}$ are represented by 20, 19, 27, 26, 25, 29, 31, 30, 21, 24, 16, 18, 17, 28, 23, 22, 13, 2, 0, 3, 9, 15, 6, 14, 7, 8, 11, 10, 1, 12, 4, 5 respectively.}
        \begin{tabular}{|*{10}{c|}}
            \hline
            No. & Op. & No. & Op. & No. & Op. & No. & Op. & No. & Op. \\
            \hline
            1 & (0,16) & 2 & (1,17) & 3 & (2,18) & 4 & (3,19) & 5 & (4,20) \\
            \hline
            6 & (5,21) & 7 & (6,22) & 8 & (7,23) & 9 & (8,24) & 10 & (9,25) \\
            \hline
            11 & (10,26) & 12 & (11,27) & 13 & (12,28) & 14 & (13,29) & 15 & (14,30) \\
            \hline
            16 & (15,31) & 17 & (29,20) & 18 & (31,20) & 19 & (20,8) & 20 & (5,7) \\
            \hline
            21 & (25,16) & 22 & (7,6) & 23 & (25,14) & 24 & (27,16) & 25 & (26,27) \\
            \hline
            26 & (8,1) & 27 & (21,11) & 28 & (19,24) & 29 & (17,26) & 30 & (9,0) \\
            \hline
            31 & (18,15) & 32 & (13,4) & 33 & (22,9) & 34 & (6,14) & 35 & (4,6) \\
            \hline
            36 & (30,31) & 37 & (21,30) & 38 & (1,9) & 39 & (6,5) & 40 & (18,19) \\
            \hline
            41 & (15,4) & 42 & (14,15) & 43 & (23,28) & 44 & (22,23) & 45 & (3,1) \\
            \hline
            46 & (14,5) & 47 & (11,3) & 48 & (0,11) & 49 & (23,30) & 50 & (1,2) \\
            \hline
            51 & (27,5) & 52 & (10,0) & 53 & (5,13) & 54 & (3,10) & 55 & (27,18) \\
            \hline
            56 & (16,12) & 57 & (23,1) & 58 & (7,5) & 59 & (18,3) & 60 & (13,15) \\
            \hline
            61 & (16,6) & 62 & (28,21) & 63 & (10,8) & 64 & (26,14) & 65 & (24,6) \\
            \hline
            66 & (17,5) & 67 & (29,22) & 68 & (31,22) & 69 & (24,17) & 70 & (12,13) \\
            \hline
            71 & (2,3) & 72 & (9,2) & 73 & (0,9) & 74 & (30,8) & 75 & (20,26) \\
            \hline
            76 & (16,30) & 77 & (8,3) & 78 & (26,12) & 79 & (4,0) & 80 & (19,26) \\
            \hline
            81 & (30,0) & 82 & (28,0) & 83 & (12,7) & 84 & (31,28) & 85 & (29,10) \\
            \hline
            86 & (23,29) & 87 & (16,19) & 88 & (19,23) & 89 & (25,19) & 90 & (21,25) \\
            \hline
            91 & (30,24) & 92 & (24,31) & 93 & (27,24) & 94 & (20,27) & 95 & (2,10) \\
            \hline
            96 & (17,4) & 97 & (16,27) & 98 & (3,12) & 99 & (27,31) & 100 & (25,18) \\
            \hline
            101 & (4,3) & 102 & (28,11) & 103 & (25,4) & 104 & (11,2) & 105 & (29,3) \\
            \hline
            106 & (20,29) & 107 & (22,25) & 108 & (17,22) & 109 & (5,12) & 110 & (1,5) \\
            \hline
            111 & (30,1) & 112 & (11,14) & 113 & (15,5) & 114 & (26,10) & 115 & (7,10) \\
            \hline
            116 & (19,0) & 117 & (22,19) & 118 & (25,23) & 119 & (18,29) & 120 & (18,13) \\
            \hline
            121 & (22,11) & 122 & (19,25) & 123 & (26,28) & 124 & (29,25) & 125 & (8,6) \\
            \hline
            126 & (15,9) & 127 & (2,13) & 128 & (28,27) & 129 & (23,8) & 130 & (23,20) \\
            \hline
            131 & (14,8) & 132 & (31,2) & 133 & (31,21) & 134 & (5,2) & 135 & (27,17) \\
            \hline
            136 & (9,11) & 137 & (17,15) & 138 & (12,0) & 139 & (0,1) & 140 & (13,4) \\
            \hline
            141 & (10,15) & 142 & (25,16) & 143 & (6,7) & 144 & (13,20) & 145 & (2,19) \\
            \hline
            146 & (0,27) & 147 & (3,26) & 148 & (9,25) & 149 & (15,29) & 150 & (6,31) \\
            \hline
            151 & (14,30) & 152 & (7,21) & 153 & (8,24) & 154 & (11,16) & 155 & (10,18) \\
            \hline
            156 & (1,17) & 157 & (12,28) & 158 & (4,23) & 159 & (5,22) & & \\
            \hline
        \end{tabular}
    \end{table}

\end{document}